\newif\ifdraft
\drafttrue 
\documentclass[12pt]{article}

\usepackage[T1]{fontenc}
\usepackage{tablefootnote}
\usepackage{lmodern, microtype}
\usepackage[left=1.25in, right=1.25in, top=1.25in, bottom=1.25in]{geometry}
\usepackage[small]{titlesec}
\usepackage{mathtools,comment,paralist}
\usepackage{epigraph}

\usepackage{amssymb,amsmath,amsthm,fancyhdr,bbm}
\usepackage{array}
\newcolumntype{L}[1]{>{\raggedright\arraybackslash}p{#1}}

\usepackage{xcolor}
\usepackage{hyperref}
\hypersetup{
  colorlinks=true,
  linkcolor=blue,
  citecolor=red
}
\usepackage{natbib}
\setcitestyle{authoryear,open={(},close={)}}
\makeatletter

\usepackage{booktabs}
\usepackage{subcaption}
\usepackage{tikz}
\usetikzlibrary{arrows.meta}
\usetikzlibrary{decorations.pathreplacing}
\usepackage{makecell}

\newtheorem{theorem}{Theorem}

\newtheorem{lemma}{Lemma}
\newtheorem*{lemma*}{Lemma}

\newtheorem*{axiom*}{Axiom}
\newtheorem{proposition}{Proposition}
\newtheorem{claim}{Claim}
\newtheorem{corollary}{Corollary}

\newtheorem*{theorem*}{Theorem}

\newtheorem{definition}{Definition}
\newtheorem*{definition*}{Definition}

\theoremstyle{definition}

\usepackage{graphicx}
\usepackage{pstricks, enumerate, pst-node, pst-text, pst-plot}
\usepackage{soul}

\usepackage{accents}

\def\Pr{\mathbb{P}}
\def\birq{\mathcal{B}_{i, r}^q}
\def\cir{C^q_{i,r}}

\def\cQ{\mathcal{Q}}
\def\eps{\varepsilon}

\title{Local Global Games and Network Common Learning}

\author{Olga Rospuskova \and Omer Tamuz \and Jake Zhang\thanks{Caltech. Olga Rospuskova was supported by a Citadel fellowship. Omer Tamuz was supported by a National Science Foundation CAREER award (DMS-1944153) and by a MURI award (N000142412742).  }}

\date{\today}

\begin{document}

\maketitle
\begin{abstract}
    We study global games in which agents coordinate locally, with their social network neighbors, contingent on a favorable state. Before acting, agents learn the private signals of all agents within network distance $r$. As $r$ grows, every agent learns the state, but efficient coordination depends on higher-order beliefs, which are shaped by the geometry of the network. We introduce network common learning, a network analogue of common learning, and show that it is attained when neighboring agents' observations differ by many signals, as on the two-dimensional grid, but fails on networks with informational bottlenecks, such as the line, where only the safe action survives in equilibrium.
\end{abstract}

\section{Introduction}

Global games are a successful and well-studied model of a number of important phenomena, including bank runs, protests and more \citep{carlsson1993global,morris2003global,goldstein2005demand,angeletos2007dynamic}. The basic model is of a group of agents that need to act in unison to achieve a goal, and furthermore can only succeed when the state of nature is favorable. The important insight they provide is that even when agents have very precise information about the state, coordination can fail because of higher-order uncertainty \citep{rubinstein1989electronic,carlsson1993global,morris2003global}. The study of global games has spurred a literature within epistemic game theory, dedicated to understanding how beliefs affect coordination. A particularly interesting concept is the idea of common learning \citep*{cripps2008common}, which captures a convergence not just to knowing the state, but to commonly knowing it \cite*[see also][]{steiner2011communication,cripps2013common,frick2023learning}.

Our setting features a classic global game. There is a state that is either high or low, and each agent has to choose between a safe and a risky action. When the state is low, the safe action is optimal. When the state is high, it is better to take the risky action, but only if others take it too. We introduce a network structure to this game, so that agents do not need to coordinate with all others, but only locally, with their social network neighbors; accordingly, we call this game a \emph{local global game}. This changes the higher-belief structure required for coordination. Instead of the common $q$-belief that is needed for coordination on a complete network, a general network requires a weaker condition of iterated beliefs, which we call network $q$-beliefs, and which we show is equivalent to the existence of an efficient equilibrium. This is similar to the conditions explored by \cite*{morris1995p} and \cite{ kets2008beliefs}.

Our network plays an additional important role, as a conduit of information. Initially, agents are endowed with conditionally independent private signals about the state. Before choosing their actions, the agents participate in a communication stage, in which they share their signals with their neighbors, their neighbors' neighbors and so on, up to some graph distance $r$, which we call the observation radius. For simplicity, and since this is a coordination game, we model this communication phase non-strategically, and assume that players simply reveal their private information.

For large $r$, each agent learns the state with high confidence. Moreover, neighboring agents observe information that is highly overlapping, seemingly allowing their beliefs to align more easily. But does individual learning and the overlap of information allow agents to resolve strategic uncertainty and coordinate efficiently?



We show that the answer to this question depends sharply on the geometry of the network. Consider two stylized examples: the infinite line and the infinite two-dimensional grid. On both networks, increasing the observation radius allows every agent to learn the state with arbitrarily high confidence. Nevertheless, the equilibrium implications are different. On the line, when coordination requires sufficiently high confidence, taking the safe action is the only action taken in equilibrium, no matter how large the radius $r$ is. On the two-dimensional grid, by contrast, when the radius is large there exist efficient equilibria in which the agents, with high probability, coordinate and take the risky action in the high state.

Why does the additional dimension reverse the equilibrium prediction? The key is not how much information each agent observes but how the available information changes across agents. Two neighboring agents observe largely overlapping neighborhoods, yet each also sees some signals that the other does not. We refer to signals observed by one agent but not by his neighbor as his \emph{information boundary}. On the line, this region contains only one signal, regardless of the size of the observation radius. On the two-dimensional grid, the information boundary grows proportionally with the observation radius.

This difference determines the feasibility of coordination. When the size of the information boundary remains bounded, there is no concentration of measure, and the uncertainty about a neighbor's information never vanishes. This concern propagates from neighbor to neighbor, eventually ruling out risky coordination.
On the other hand, when the information boundary is large, the unobserved signals become predictable, because of concentration of measure. Strategic uncertainty diminishes as the radius grows, allowing confidence in risky play to sustain itself across the network. 

To construct efficient equilibria, we study what we call \emph{network $q$-belief} \citep*[see also][]{morris1995p,kets2008beliefs}. This notion is defined recursively. At the first level, an agent is required to assign probability at least $q$ to the state being high. At each subsequent level, he needs to assign probability at least $q$ to the joint event that the state is high and that all his neighbors assign probability $q$ to the events defined at the preceding level. We show that the limit of this iterated belief satisfies the same fixed-point condition as the agents' best responses in  equilibrium. Namely, it yields a collection of \emph{self-sustaining events}: for each agent, whenever his corresponding event in the collection happens, he assigns probability at least $q$ to the joint event that the state is favorable and that every neighbors' corresponding event also happens. In analogy to the idea of common learning \citep*{cripps2008common}, we define \emph{network common learning} to hold if for all $q$, network $q$-belief holds with high probability for all $r$ large enough.

We use this technique to prove that coordination is possible on the two  dimensional grid. The key geometric property of the grid is what we call \emph{large information boundaries} (Definition~\ref{def:boundaries}), a condition that applies to many other networks. To define the large information boundaries condition, we first define the collection of \emph{test regions}: these are regions in the social network graph that are finite intersections of the agents' observation neighborhoods. Large information boundaries are attained when every large test region that is observed by some agent $i$ has a large information boundary: the part of it that is not seen by $i$'s neighbor $j$ is large (or empty). The term ``large'' is used informally here, but is defined formally below.

Large information boundaries allow us to define self-sustaining events that yield efficient equilibria. The self-sustaining events require that in each test region the fraction of favorable signals is large. That is, to coordinate an agent needs to not just see many favorable signals (yielding individual learning) but see many favorable signals in each test region, which is an intersection of the regions seen by a group of agents.


Our final result provides a complementary impossibility theorem. Suppose the network contains a ``bottleneck'': a path along which moving from one agent to the next introduces a bounded number of new signals; this holds for the one-dimensional grid. That is, the information boundary of the observation neighborhoods along this path is of bounded size. This implies that the higher-order uncertainty generated by the boundary is significant regardless of $r$, and is not mitigated by concentration of measure. Consequently, when $q$ is high enough, the collection of events that can be believed with probability $q$ shrinks as we iterate over higher-order beliefs and move along the path, diminishing to the empty set if the path is long enough. Network $q$-belief is consequently empty, and the safe action is the unique equilibrium prediction.

\medskip

Our paper identifies a new channel through which network geometry affects coordination. Network geometry matters not only because it determines who interacts with whom and how information is diffused, but also because it determines how much uncertainty remains when agents reason about the actions of their neighbors. 
When neighboring observation sets differ by many signals, the fresh information received by each agent becomes predictable in the aggregate, allowing confidence in others’ actions to be sustained. When observation sets change only gradually, uncertainty about a neighbor’s information cannot be ignored and may propagate over long distances, even though every individual agent is highly informed about the state.

\subsection{Related literature}
Our model follows the literature on global games \citep{carlsson1993global, morris2003global}.  The basic lesson of this literature is that arbitrarily precise private information need not induce coordination because equilibrium behavior also hinges on higher-order beliefs \citep{rubinstein1989electronic}. We inherit this lesson but modify the object of higher-order beliefs. In the canonical model, an agent needs to coordinate with the entire population, and this imposes conditions on belief involving every agent at every epistemic order. In our model, an agent's payoff only depends on the actions of his neighbors, and so the relevant object is network $q$-belief, which chains beliefs along the edges of the network. Despite this difference from the classic literature, our impossibility result is in the same family as the coordinated attack problem of \citet{halpern1984knowledge}: the argument in the proof of 
Theorem~\ref{thm:path-impossibility} is a spatial version of the temporal mechanism governing the electronic mail game of \citet{rubinstein1989electronic}, with graph distance along a
bottleneck path playing the role of the number of messages exchanged.

Network q-belief is a network analogue of common p-belief, as defined by \citet{monderer1989approximating}, who show that common p-belief admits a fixed-point characterization; we derive a network analogue of their fixed-point property. Common p-belief is also central to the analysis of \citet{kajii1997robustness}, who show that an equilibrium is robust to incomplete information when the event on which the game is common knowledge can be commonly p-believed, and whose contagion arguments are echoed in our impossibility result. \cite{kets2008beliefs} is closest to us in adapting the belief operator to the network structure. Beyond uncertainty about the state, \cite{kets2008beliefs} also considers uncertainty about the network structure, and applies these techniques to the study of anonymous network games.

\cite{morris2000contagion} studies how local interactions determine whether behavior can spread contagiously through a network. In his complete information setting, agents repeatedly best respond to the actions of their neighbors. He shows that an action adopted by a finite group of agents can be spread through the entire network only when there is no cohesive group that interacts mostly within the group and resists the outside contagion. Our analysis shares the emphasized propagation and recursive structure but the object being propagated is different. In Morris, an agent counts the proportion of his neighbors taking a certain action and the observed action diffuses through the network. In contrast, we study how uncertainty propagates through higher-order beliefs and the network structure. Network geometry therefore matters not only through who interacts with whom, as in \cite{morris2000contagion}, but also through how information changes across neighboring positions.
Related work on local interaction dynamics and geometry includes \citet{ellison1993learning}, \citet{blume1993statistical} and
\citet{oyama2015contagion}.

\cite*{cripps2008common} study agents who accumulate private signals over time and ask whether this leads to common $p$-belief. They show that when the number of signal realizations is finite, individual learning implies common learning; see also \cite*{frick2023learning}. Our observation radius $r$ plays the role of time in their model: as the radius $r$ grows, the information of each agent becomes arbitrarily precise and the question of interest is whether higher-order beliefs also grow. In our model, individual learning does not always imply common learning, even with finite signals; it is the geometry of the graph which decides.

\citet{chwe2000communication} also studies coordination in which a network is used to communicate, and asks which networks are minimally sufficient for coordination. In his model communication is used to build exact common knowledge within the set of participants, so the binding constraint is whether a network has the right links.

A recent paper by \cite*{hutchcroft2026local} also studies coordination on social networks using local communication, but in a crucially different setting, resulting in very different conclusions. In that paper, there is no global state, and agents play a pure coordination game (e.g., choosing which side of the road to drive on). Thus, the challenge is to communicate and coordinate on how to break ties between a priori equivalent actions. In that setting, coordinating on the line graph is possible, and in fact higher efficiency can be achieved on the line than on the two dimensional grid. Amenability, the geometric property enabling coordination in that paper, is related to having small information boundaries, while in our setting large information boundaries are a sufficient condition for coordination.

\section{Model}

Let $N$ be a set of agents, connected through a social network $G=(N, E)$.  An edge $(i, j) \in E$ indicates that agents $i$ and $j$ are neighbors. We assume that the social network is undirected, so that $(i, j)\in E$ if and only if $(j, i)\in E$. For convenience, $(i, i) \not \in E$. For each $i\in N$, let $N_i = \{j\in N: (i, j) \in E\}$ denote the set of $i$'s neighbors. We assume that every agent has finitely many neighbors, i.e. $d_i = |N_i|< \infty$ for every $i\in N$. We will be interested in the limit of very large networks, and so for simplicity we will assume that the number of agents is countably infinite.

A path from $i$ to $j$ is a finite sequence of agents beginning at $i$ and ending at $j$, such that every two consecutive agents are neighbors. We assume that $G$ is connected, i.e., that there is a path connecting each pair of agents. This is without loss of generality, since otherwise one can consider each connected component in isolation. The distance between $i$ and $j$ is the length of the shortest path between them. Each agent observes all agents at graph distance at most $r$, where $r\geq 0$. We call $r$ the observation radius. We also define $B_r(i)\subseteq N$ as the set of agents at distance at most $r$ from $i$. We refer to this set as $r$-neighborhood of $i$, or as the observation neighborhood of agent $i$. 

We consider a coordination game played on the network $G$. 
The payoff of an agent depends on the underlying state, his own action and the actions of his neighbors. The state is binary: $\theta \in \Theta$, where $\Theta = \{H, L\}$, with $H$ and $L$ denoting the high and low states respectively.
Each agent chooses between a safe action $S$ and a risky action $R$. The safe action always guarantees a payoff of zero. The risky action yields a payoff $\alpha > 0$ only if the state is $H$ and all his neighbors coordinate on the risky action; otherwise, it incurs a loss of $\beta > 0$. Formally, the payoff of an agent $i$ under state $\theta$ and action profile $a \in \{S,R\}^N$ is given by
\begin{align*} u_i(a,\theta) = \begin{cases}
    0, & \text{ if } a_i = S,\\
    \alpha, & \text{ if } \theta = H, a_i = R \text{ and } a_j = R \ \forall j\in N_i, \\
    -\beta, & \text{ otherwise.}
\end{cases} \end{align*}



The agents are uncertain about the state and share a common prior. For simplicity, we assume $\Pr (\theta = H) = \Pr(\theta = L) = 1/2$. Each agent $i$ receives a binary signal $s_i$ about the state. Conditional on the state, the signals are independent. We assume that the signals are symmetric with precision $p > 1/2$ and take values in $\{0,1\}$.  That is, \[
\Pr(s_i = 1|\theta = H) = \Pr(s_i = 0|\theta=L) = p.
\]
The state and signals are random variables in the probability space $(\Omega,\mathcal{F},\Pr)$, where $\Omega = \{H,L\}\times\{0,1\}^N$, $\mathcal{F}$ is the product sigma-algebra, and $\Pr$ is the measure described above.

Before choosing an action, agents participate in a communication phase. In this phase, each agent $i$ observes the private signals of all agents in $B_r(i)$, for some fixed radius of observation $r$. We denote the set of possible observations by $O^r_i = \{0,1\}^{B_r(i)}$. The information available to agent $i$ is the random variable $I_i^r = (s_j)_{j \in B_r(i)}$, which takes values in $O_i^r$. Thus, a pure strategy for agent $i$ in this game is a map $t_i \colon O_i^r \to \{S,R\}$. This is a Bayesian game, and accordingly our solution concept is that of Bayes-Nash equilibrium.

Since the network is infinite and connected, the size of $B_r(i)$ goes to infinity as $r$ increases, and so the number of signals observed by each agent in the communication phase increases. Since private signals are informative about the state, agents individually learn the state. In particular, as $r$ tends to infinity, each agent's posterior belief $\Pr(\theta=H|I_i^r)$ converges to $1$ if $\theta=H$, and converges to $0$ if $\theta=L$, i.e. for all $\eps>0$,
\begin{align*}
    \lim_{r \to \infty}\Pr(\Pr(\theta=H|I_i^r)>1-\eps|\theta=H) = \lim_{r \to \infty}\Pr(\Pr(\theta=H|I_i^r)<\eps|\theta=L) =1.
\end{align*}

We focus on pure Bayes-Nash equilibria, which are characterized by a collection of sets $(X_i^r)_{i\in N}$, where $X_i^r\subseteq O_i^r$. These are the sets on which $i$ chooses the risky action, i.e.,  $a_i = R$ if and only if $I_i^r \in X_i^r$. The expected utility of agent $i$ from choosing the risky action upon observing $I_i^r = x$ is
$$\alpha \Pr(\theta = H, I_j^r \in X_j^r \ \forall j \in N_i|I_i^r = x) - \beta (1 - \Pr(\theta = H, I_j^r \in X_j^r \ \forall j \in N_i|I_i^r = x)),$$
whereas choosing a safe action yields zero. 

Therefore, the strategy profile induced by $(X_i^r)_{i\in N}$ is an equilibrium if and only if, for every $i\in N$, 
\begin{align*}
    & \forall x \in X_i^r, \quad \Pr(\theta =H,  I_{j}^r \in X_j^r \ \forall j\in N_i |  I_{i}^r = x)\geq \frac{\beta}{\alpha + \beta}, \\
    & \forall y \notin X_i^r, \quad \Pr(\theta =H, I_{j}^r \in X_j^r \ \forall j\in N_i |  I_{i}^r = y) \leq \frac{\beta}{\alpha + \beta}.
\end{align*}

Denoting $q =\frac{\beta}{\alpha + \beta}$, the risky action is optimal whenever the probability that the state is high and all neighbors choose risky action is at least $q$. We resolve ties in favor of $R$, assuming that agents take the risky action when they are indifferent
.

We ask for which networks agents can coordinate efficiently in equilibrium when they observe signals within a sufficiently large radius $r$. That is, with high probability, they choose the risky action in the high state and the safe action in the low state.

\subsection{Network $q$-belief and Network Common Learning}

For any event $E\subseteq\Omega$, define $$\birq(E) = \{\omega\in \Omega: \Pr(E|I_i^r)(\omega) \geq q\}.$$

This is the event on which agent $i$, using signals in $B_r(i)$, assigns probability at least $q$ to $E$.\footnote{Since $B_r(i)$ is finite, $\sigma(I_i^r)$ is generated by the finite partition $\big\{\{I_i^r = x\}\,:\, x \in O_i^r\big\}$, in which every atom has positive probability. Hence the conditional probability $\Pr(E | I_i^r)$ admits a canonical version, given on each atom by $\Pr(E \cap \{I_i^r = x\})/\Pr(I_i^r = x)$, and  $\birq(E)$ is well defined as the union of the atoms on which this value is at least $q$.} If each neighbor $j$ of $i$ chooses the risky action on $X_j^r$, then $\birq(H\cap (\bigcap_{j\in N_i} X_j^r))$ is the event on which agent $i$ assigns a probability at least $q$ to the event that the state is high and every neighbor of $i$ chooses $R$.\footnote{Here and below we use $H$ to denote the event $\{\theta =H\}$, and use $X_j^r$ to denote the event $\{I_j^r \in X_j^r\}$.} Therefore, the equilibrium condition can equivalently be written as \begin{align} \label{eq:eqm}
     \forall i\in N, \quad X_i^r = \birq\left(H\cap \left(\bigcap_{j\in N_i}X_j^r\right)\right).
\end{align}

This is a fixed-point condition. To identify its solutions, we study iterated network beliefs. Starting from the entire sample space $\Omega$, at each subsequent step agent $i$ leaves only those in which he $q$-believes that state is $H$ and that the previous-level condition holds for every neighbor. Formally, let $C^{q,0}_{i,r} = \Omega$. For $n\geq 1$, define
\begin{align*}
 C^{q,n}_{i,r}(H) = \birq\left(H\cap \left(\bigcap_{j\in N_i}C_{j,r}^{q, n-1}(H)\right)\right).
\end{align*}
This is the event that agent $i$ has $n$\textsuperscript{th} order network $q$-belief. In particular, $$C^{q,1}_{i,r}(H) =\birq(H) $$ is the event in which agent $i$ assigns probability at least $q$ to the high state. Likewise, $$C^{q,2}_{i,r}(H) = \birq\left(H\cap \left(\bigcap_{j\in N_i}\mathcal{B}_{j, r}^q(H)\right)\right)$$ is the event that agent $i$ has second order network $q$-belief: $i$ assigns probability at least $q$ to the joint event that state is $H$ and all his neighbors believe that state is $H$ with probability at least $q$. At higher levels, the same condition is applied recursively to $i$'s neighbors, their neighbors and so on. 

 
 The event that $i$ has network $q$-belief is 
 $$\cir(H) =\bigcap_{n\geq 0} C^{q,n}_{i,r}(H) .$$
This is the event on which $q$-belief in $H$ persists through every finite order of beliefs, starting from $i$. The following lemma shows that the limiting event satisfies the same fixed-point condition required for equilibrium. To simplify notation, we henceforth omit the $(H)$ and write $C^{q,n}_{i,r}$ and $\cir$ to indicate $C^{q,n}_{i,r}(H) $ and $\cir(H)$.

\begin{lemma}
\label{lem:fixedpoint}
    For every $i$ and $r$,
    $$\cir  = \birq\left(H\cap \left(\bigcap_{j\in N_i}C^q_{j,r}\right)\right).$$
\end{lemma}

Thus, network $q$-belief is \textit{self-sustaining}: whenever it holds at $i$, agent $i$ assigns probability at least $q$ to the joint event that the state is $H$ and network $q$-belief holds at each of $i$'s neighbors.

Since the family of sets $(C^{q}_{i,r})_{i \in N}$ satisfies the equilibrium condition \eqref{eq:eqm}, it induces an equilibrium. The next proposition shows that this equilibrium has the largest possible risky set: whenever another equilibrium prescribes the risky action to agent $i$, this equilibrium does so as well.


\begin{proposition}
\label{prop:maximal_eqm}
    The family $(\cir)_{i\in N}$ induces an equilibrium in which $X_i^r = \cir$. Moreover, it is the maximal risky equilibrium in the sense that in any other equilibrium, $a_i=R$ only on the event $\cir$.
\end{proposition}

We now study whether network $q$-belief in $H$ emerges as agents observe larger neighborhoods. Particularly, this requires that events $(\cir)_{i\in N}$ become very likely for every agent in state $H$. This motivates the following notion of network common learning. 

\begin{definition}
    Network $q$-learning of $H$ occurs if for every agent $i$, 
    $$  \lim_{r \to \infty}\Pr(\cir|\theta = H) =1.$$
    Network common learning of $H$ occurs if network $q$-learning of $H$ occurs for all $q \in (0,1)$.
\end{definition}

So, network $q$-learning means that in state $H$ network $q$-belief in $H$ occurs with probability approaching $1$ for every agent. It also implies asymptotic efficiency of the maximal risky equilibrium. Since the maximal equilibrium prescribes the risky action on these events, risky play becomes asymptotically likely in state $H$.  Individual learning implies that in state $L$ each agent chooses the safe action with high probability. For choosing the safe action, no higher-order beliefs are necessary: once the first-order belief in $H$ is below $q$, agents will choose the safe action.


\section{Efficient Coordination on Networks with Large Information Boundaries}

In this section we introduce a geometric condition on the social network graph which we show implies that network common learning is obtained, and thus efficient coordination can be sustained in equilibrium, for all $r$ large enough.

To define our geometric condition, recall that $B_r(i)$ is the ball of radius $r$ around agent $i$. Denote by $\cQ^r$ the family of all subsets of the agents that are (finite) intersections of balls of radius $r$:
\begin{align*}
    \cQ^r = \{ \cap_{k=1}^m B_r(i_k)\,:\, m \geq 1, i_1,\ldots,i_m \in N\}.
\end{align*}
We call each $Q \in \cQ^r$ a \emph{test region}, as these will be the sets of signals that will determine if players coordinate.
On the line graph, test regions are simply the finite intervals of length at most $2r+1$. On the two dimensional grid, where balls are diamond-shaped, these are 45-degree slanted rectangles (see Figure~\ref{fig:grid-information-boundary}). On any graph, a test region $Q$ includes those agents whose signals are observed by some finite group of agents $i_1,\ldots,i_m$, which is the reason they are useful in constructing equilibria that depend on these agents' higher-order beliefs.

We denote by $\cQ_i^r \subset \cQ^r$ the test regions that are included in $B_r(i)$:
\begin{align*}
    \cQ^r_i = \{Q \in \cQ^r\,:\, Q \subseteq B_r(i)\} = \{ Q \cap B_r(i) \,:\, Q \in \cQ^r\}.
\end{align*}
These are the test regions that are observed by $i$. 

Suppose $i$ and $j$ are neighbors in the graph. Given a test region $Q \in \cQ_i^r$ (observed by $i$), we call $Q \setminus B_r(j)$ the $(i,j)$-\emph{information boundary} of $Q$ (see Figure~\ref{fig:grid-information-boundary}). This is the part of $Q$ that is not observed by $j$. Our condition of \emph{large information boundaries} requires that when $Q$ is large, its  information boundaries are large, or else the boundaries are empty. 
\begin{definition}
    \label{def:boundaries}
    The social network graph has \emph{$(\eps,\delta)$-large information boundaries} if for all $r>0$ and neighbors $i,j$ it holds that for every test region $Q \in \cQ^r_i$ of size at least $\eps|B_r(i)|$, the $(i,j)$-information boundary $Q \setminus B_r(j)$ is either empty, or else is of size at least $\delta \cdot r$.

    The social network graph has \emph{large information boundaries} if for all $\eps>0$ there is a $\delta>0$ such that the graph has $(\eps,\delta)$-large information boundaries.
\end{definition}
In the case in which the boundary is empty, higher-order beliefs are trivial: $Q$ is completely observed by $j$. When not all of $Q$ is observed by $j$, this property requires that the boundary is large, which will help us control $j$'s higher-order beliefs.

\begin{figure}[t]
\centering

\begin{tikzpicture}[
    x=0.72cm,
    y=0.72cm,
    vertex/.style={circle,fill=black,inner sep=1.5pt},
    bigvertex/.style={circle,fill=black,inner sep=2.3pt},
    boundaryvertex/.style={
        circle,
        fill=red,
        draw=red,
        inner sep=2.3pt
    },
    font=\scriptsize
]

\fill[red!8]
    (0,5) -- (2,3) -- (-3,-2) -- (-5,0) -- cycle;

\foreach \x in {-9,...,9} {
    \draw[line width=0.35pt] (\x,-6) -- (\x,6);
}
\foreach \y in {-6,...,6} {
    \draw[line width=0.35pt] (-9,\y) -- (9,\y);
}

\foreach \x in {-9,...,9} {
    \foreach \y in {-6,...,6} {
        \node[vertex] at (\x,\y) {};
    }
}

\draw[black!55,dashed,line width=0.9pt]
    (1,5) -- (6,0) -- (1,-5) -- (-4,0) -- cycle;

\draw[red,line width=1.1pt]
    (0,5) -- (2,3) -- (-3,-2) -- (-5,0) -- cycle;

\draw[red!55,line width=3pt]
    (0,5) -- (-5,0) -- (-3,-2);

\foreach \x/\y in {
     0/5,
    -1/4, 
    -2/3, 
    -3/2, -3/-2,
    -4/1, -4/-1,
    -5/0
} {
    \node[boundaryvertex] at (\x,\y) {};
}

\node[bigvertex] at (0,0) {};
\node[bigvertex] at (1,0) {};

\node[below=6pt] at (0.2,0.3) {$i$};
\node[below=6pt] at (1.2,0.3) {$j$};

\node[red,fill=white,inner sep=1.5pt] at (1.35,3.15) {$Q$};

\node[black!60,fill=white,inner sep=1.5pt]
    at (3.85,2.15) {$B_r(j)$};

\node[red,fill=white,inner sep=1.5pt]
    (boundarylabel) at (-6.8,4.6)
    {$Q\setminus B_r(j)$};

\draw[red,->,line width=0.8pt]
    (boundarylabel.south east) -- (-3.2,2.3);

\end{tikzpicture}

\caption{A test region $Q \subset B_r(i)$ is shown in red. The observation neighborhood of $j$ is shown by a dotted line. The highlighted red agents form the
nonempty $(i,j)$-information boundary $Q\setminus B_r(j)$.}
\label{fig:grid-information-boundary}
\end{figure}
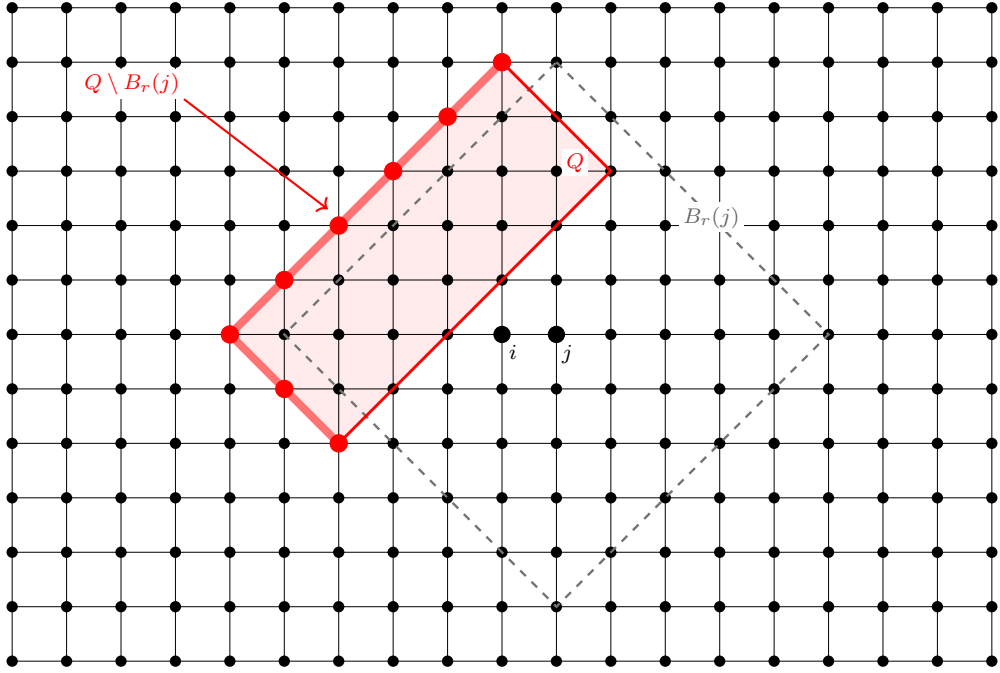

For network common learning we will need two additional geometric conditions. The first is that the collection of test regions $\cQ_i^r$ is not too large.
\begin{definition}
    The social network graph has \emph{subexponential}-$\cQ$ if the size of the collection $\cQ_i^r$ satisfies $\lim_{r \to \infty} \left(\sup_{i \in N}|\cQ_i^r|\right)^{1/r}=1$.
\end{definition}
The second is that the size of neighborhoods does not vary too much between agents.
\begin{definition}
    The social network graph has \emph{similar neighborhood sizes} if there is a constant $K \geq 1$ such that for all $r>0$ and for all $i,j \in N$ it holds that $|B_r(i)|/|B_r(j)| \leq K$.
\end{definition}
It is straightforward to verify that all three of these conditions are satisfied by the two-dimensional grid.\footnote{In the coordinates $u = x+y$, $v = x-y$, the ball $B_r(i)$ of
$\mathbb{Z}^2$ is a box $\{|u-u_i|\leq r\} \cap \{|v-v_i|\leq r\}$, so test
regions are boxes $[u_1,u_2]\times[v_1,v_2]$. Similar neighborhood sizes holds
with $K=1$ since the graph looks the same everywhere. A test region contained in $B_r(i)$ is
determined by four coordinates in a window of length $2r+1$, so
$|\mathcal{Q}_i^r| \leq (2r+1)^4$, giving subexponential-$\mathcal{Q}$. For
large boundaries, note that $|B_r(i)| \geq 2r^2$, so a box $Q$ with
$|Q|\geq \eps |B_r(i)|$ and side lengths at most $2r+1$ has both sides of
length at least $\eps r/2$. Moving to a neighbor $j$ shifts the box $B_r(j)$
by one unit in $u$ and in $v$; hence if $Q \setminus B_r(j)$ is nonempty, $Q$
protrudes past a face of $B_r(j)$, and the boundary contains a full side of
$Q$, of size at least $\delta r$ with $\delta = \eps/4$.} Moreover, they are satisfied by higher dimensional grids, but much more generally by the large collection of graphs that have the same global geometry as these graphs, in precisely the sense captured by these definitions.

Our main theorem is that these conditions are sufficient for network common learning.
\begin{theorem}
    \label{thm:common-learning}
    Suppose that $G$ satisfies the large information boundary condition,  the subexponential-$\cQ$ condition, and the similar neighborhood sizes condition. Then network common learning is obtained on $G$. 
\end{theorem}
The intuition behind the proof of this result is the following. Large information boundaries imply that when agent $j$ reasons about his neighbor $i$'s belief, and in particular when he reasons about the signals $i$ observed in a test region $Q$, a large information boundary allows for concentration of measure to reduce the uncertainty $j$ has about the aggregate signal delivered by the unobserved signals.

In particular, to show that network $q$-belief is obtained, we consider the event that within the ball $B_r(i)$ observed by $i$, not only are there many high signals (indicating that the state is most likely high) but moreover in each test region $Q \in \cQ^r_i$ separately there are many high signals. When the collection of test regions $\cQ^r_i$ is small, we show that this happens with high probability. More importantly, this collection of events is self-sustaining, leading to network $q$-belief. We explain this construction in further detail here, and defer the full proof to the appendix.

Let $M_r = \sup_j |B_r(j)|$ be the size of the largest neighborhood of radius $r$. The similar neighborhood sizes property ensures that this is finite for each $r$. Given a finite subset $Q$ of the agents, let $S(Q) = \sum_{j \in Q}s_j$ be the number of high signals observed by the agents in $Q$. Choosing $\eps>0$ small enough, the self-sustaining events $Y^r_i$ are given by
\begin{align}\label{eq:self-sustaining-events}
Y^r_i = \left\{ S(Q) \geq (p-\eps)|Q|-\eps M_r
 \quad \forall Q\in \mathcal{Q}_i^r
\right\}. 
\end{align}
That is, there are sufficiently many high signals in every test region observed by $i$.


Note that in the high state, the expectation of $S(Q)$ is $p|Q|$, which is larger than $(p-\eps)|Q|$, and so for large $Q$ this condition holds with high probability. Note also that if $Q$ is small---in particular smaller than $\eps M_r$---then this condition holds trivially, since the right hand side of the inequality is negative. Using the fact that the size of the collection $\cQ^r_i$ is small, we show that $Y^r_i$ has high probability. 

The crux of the proof is to show that the collection $Y^r_i$ is self-sustaining, i.e., that 
$$Y_i^r \subseteq \birq \left (H\cap \left(\bigcap_{j\in N_i} Y_j^r\right)\right).$$
For this, we use the assumption that $G$ has large information boundaries.

\section{Coordination Failure on  Networks with Bottlenecks}

In this section we study graphs on which the only equilibrium is the one in which agents always play the safe action, for any radius of observation. We show that a geometric condition on the graph is sufficient to guarantee this: the graph must contain a long path along which information boundaries are of uniformly bounded size. On such graphs, network common learning is not attained.

\begin{definition}
    Fix $r,m > 0$. A path of distinct agents $i_1, i_2,  \ldots, i_T$ in the graph $G$ is an $(r, m)$-bottleneck if 
    \begin{enumerate}
       \item $|B_r(i_{t+1})\backslash B_r(i_t)|\leq m$ for all $t < T$, i.e., agent $i_{t+1}$ has no more than $m$ agents in their $r$-neighborhood that are not in agent $i_t$'s $r$-neighborhood.
       \item $B_r(i_1) \cap B_r(i_T) = \emptyset$, i.e., the first and last agent's $r$-neighborhoods are disjoint.
    \end{enumerate}
\end{definition}

The first condition in the definition is that the information boundary  $B_r(i_{t+1}) \setminus B_r(i_t)$ is of size at most $m$, i.e., that along the path, $r$-neighborhoods are highly overlapping, differing by at most $m$ agents. The second condition ensures that the path is long enough so that all signals observed by the first agents are not observed by the last; equivalently, that $i_1$ and $i_T$ are more than distance $2r$ apart. A simple example is the line graph, which contains an $(r,1)$-bottleneck for all $r$. 


Our main result in this section is that on graphs with bottlenecks coordination is impossible. 
\begin{theorem}\label{thm:path-impossibility} Suppose that $G$ contains an $(r, m)$-bottleneck. If $q > 1- (1-p)^m$, then $\Pr(C^q_{i,r})=0$ for all $i$ (i.e., network $q$-belief does not occur) and thus in the unique equilibrium all agents choose the safe action. 
\end{theorem}
The proof of Theorem~\ref{thm:path-impossibility} appears in the appendix. The idea behind it is the following argument, which is similar to infection arguments used in the literature \citep[e.g.,][]{carlsson1993global}.

\begin{figure}[t]
\centering

\begin{tikzpicture}[
    x=0.8cm,
    y=0.8cm,
    vertex/.style={circle,fill=black,inner sep=1.5pt},
    bigvertex/.style={circle,fill=black,inner sep=2.3pt},
    font=\scriptsize
]

\draw[thick] (0,0) -- (17,0);

\foreach \x in {0.4,0.8,...,16.8} {
  \node[vertex] at (\x,0) {};
}

\node[draw=none,fill=none] at (8.4,0) {$\cdots$};

\coordinate (i)  at (3.2,0);
\coordinate (j)  at (13.6,0);

\node[bigvertex] at (i)  {};
\node[bigvertex] at (j)  {};

\node[below=6pt] at (i)  {$i$};
\node[below=6pt] at (j)  {$j$};

\draw[red,line width=1.1pt]
  (1.35,0.58) -- (5.05,0.58);
\draw[red,line width=1.1pt]
  (1.35,0.58) -- (1.35,0.16);
\draw[red,line width=1.1pt]
  (5.05,0.58) -- (5.05,0.16);

\node[red] at (5.05,0.84) {$x_i$};

\foreach \x/\t in {
  1.6/0, 2.0/1, 2.4/1,
  2.8/0, 3.2/1, 3.6/0, 4.0/1, 4.4/1, 4.8/0
} {
  \node[red] at (\x,0.3) {$\t$};
}

\draw[red,line width=1.1pt]
  (11.40,0.68) -- (15.00,0.68);
\draw[red,line width=1.1pt]
  (11.40,0.68) -- (11.40,0.26);
\draw[red,line width=1.1pt]
  (15.00,0.68) -- (15.00,0.26);

\node[red] at (11.40,0.94) {$x_{j-1}$};

\draw[red,line width=1.1pt]
  (11.80,0.58) -- (15.40,0.58);
\draw[red,line width=1.1pt]
  (11.80,0.58) -- (11.80,0.16);
\draw[red,line width=1.1pt]
  (15.40,0.58) -- (15.40,0.16);

\node[red] at (15.50,0.84) {$x_j$};

\foreach \x in {12.0,12.4,12.8,13.2,13.6,14.0,14.4,14.8,15.2} {
  \node[red] at (\x,0.30) {$0$};
}
\node[red] at (11.6,0.30) {$1$};


\end{tikzpicture}
\caption{Proof sketch of Theorem~\ref{thm:path-impossibility}. The observation $x_i$ is arbitrary, while $x_j$ is fixed to induce a low belief. Conditioned on $j-1$ observing $x_{j-1}$, there is sufficiently high probability that $j$ observed $x_j$, ruling out network $q$-belief by $j-1$. By induction, this rules out network $q$-belief by $i$, for arbitrary $x_i$.\label{fig:sketch}}
\end{figure}
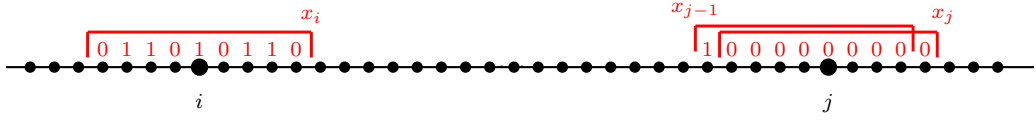

For this proof sketch, consider the one-sided line graph, i.e., $N = \mathbb{N}$ and $i,j \in N$ are neighbors if $|i-j|=1$. Fix an $r>0$ and an agent $i$, and let $x_i \in O_i^r$ be an arbitrary observation by agent $i$. To prove the claim, we show that after observing $x_i$, agent $i$ does not have network $q$-belief, and therefore must take the safe action. 

To this end, consider some agent $j$ who is far enough from agent $i$ so that their observation neighborhoods are disjoint: $B_r(i) \cap B_r(j) = \emptyset$ ($j=i+2r+1$ suffices). Let $x_j \in O_j^r$ be any observation by $j$ which implies that $j$ does not have network $q$-belief. For example, this holds for any $x_j$ which induces a posterior belief that is lower than $q$:  $\Pr(\theta=H|I_j^r=x_j) < q$. Fix any signal realizations for the agents $i+1,i+2,\ldots,j-1$ that are compatible with $x_i$ and $x_j$, and denote the observations by these agents by $x_{i+1},\ldots,x_{j-1}$; see Figure~\ref{fig:sketch}. 

The key observation is that because $x_j$ does not induce network $q$-belief for agent $j$, from agent $j-1$'s point of view, this $(I_j^r=x_j)$ happens with probability at least $1-p$: only one signal observed by $j$ is not observed by $j-1$, and so the probability that $I_j^r=x_j$ conditioned on $I_{j-1}^r=x_{j-1}$ is bounded away from zero, regardless of how large $r$ is. It follows that $x_{j-1}$ does not induce network $q$-belief in $j-1$. By induction, $x_i$ does not induce network $q$-belief in agent $i$. Since $i$ and $x_i$ are arbitrary, this completes the proof.

\bigskip

As mentioned above, the line graph has a uniform family of bottlenecks, i.e., an $(r,1)$-bottleneck for all $r$. An immediate corollary of Theorem~\ref{thm:path-impossibility} is that the existence of such a uniform family excludes network common learning.
\begin{corollary}
\label{cor:impossibility-path}
    Suppose that there exists an  $m > 0$ such that for every $r > 0$ the graph contains an $(r, m)$-bottleneck. Then network common learning is not obtained.   
\end{corollary}

\section{Conclusion}

We study global games, played locally on a social network. Agents receive a large amount of information about the state by learning the private signals of agents in their $r$-neighborhood. As in the classic literature, we ask whether higher-order beliefs become aligned. In particular, we study network common learning, the network analogue of common learning. The answer in our case is that it depends non-trivially on the network geometry.

We do not provide a full characterization of graphs on which network common learning is obtained, instead providing separate necessary and sufficient conditions. A major point for future research is to close this gap. A natural question is whether network common learning is a \emph{quasi-isometry} invariant. That is, is network common learning robust to local changes in the graph that do not change the global geometry?\footnote{Two connected graphs $G_1=(V_1,E_1)$ and $G_2=(V_2,E_2)$ are isometric (or isomorphic) if there is a bijection $\varphi \colon V_1 \to V_2$ such that $d_1(i,j) =d_2(\varphi(i),\varphi(j))$,
where $d_i$ denote graph distances. They are quasi-isometric if there is a map $\varphi \colon V_1 \to V_2$ and constants $A \geq 1$, $B\geq 0$ such that for every $i,j \in V_1$ 
\begin{align*}
    \frac{1}{A}d_1(i,j)-B \leq d_2(\varphi(i),\varphi(j)) \leq A d_1(i,j)+B,
\end{align*}
and every element of $V_2$ is of distance at most $B$ from the image of $\varphi$.}

Our model focuses on a binary state and binary signals. We believe that our network common learning results extend to any finite number of states and finite number of signals. The case of infinitely many signals is more mysterious, as it is in the context of common learning.

\newpage

\appendix
\setcounter{lemma}{0} 
\renewcommand{\thelemma}{\thesection.\arabic{lemma}}
\setcounter{claim}{0} 
\renewcommand{\theclaim}{\thesection.\arabic{claim}}

\section{Proofs}
\subsection{Proof of Proposition \ref{prop:maximal_eqm}}
\label{pf:maximal_eqm}
\begin{proof}
Firstly, note that by Lemma \ref{lem:fixedpoint} the family $(\cir)_{i\in N}$ satisfies the equilibrium condition \eqref{eq:eqm}. Thus, it induces an equilibrium. 

Next, we prove by mathematical induction that for every agent $i$,  $X_i^r \subseteq C_{i, r}^{q,n}$ for every $n$ and thus $X_i^r \subseteq \bigcap_{n \geq 0} C^{q,n}_{i,r}  = \cir$.

For $n=0$, $X_i^r \subseteq \Omega = C^{q,0}_{i,r}$. Now suppose that for every agent $j$, $X_j^r \subseteq C_{j, r}^{k}$ for $k \leq (n -1)$. Then, $H\cap \left(\bigcap_{j\in N_i}X_{j}^r\right) \subseteq H\cap \left(\bigcap_{j\in N_i}C_{j,r}^{n-1}\right)$. Since the belief operator is monotone, 
$$X_i^r = \birq\left (H\cap \left(\bigcap_{j\in N_i}X_{j}^r\right)\right) \subseteq \birq\left(H\cap \left(\bigcap_{j\in N_i}C_{j,r}^{q,n-1}\right)\right) = C^{q,n}_{i,r}.$$

Finally, consider any mixed strategy $\sigma_i : I_i^r \to [0,1]$ with $\sigma_i(x)$ denoting the probability of choosing action $R$ on observing $x$ and denote $Y=\{x \in O_i^r: \sigma(x) >0\} $ as the observation on which action $R$ is played with positive probability, we know $Y \subseteq \birq\left (H\cap \left(\bigcap_{j\in N_i}X_{j}^r\right)\right) $ and the proof for the case of mixed strategy is essentially the same as above.
\end{proof}

\subsection{Proof of Theorem~\ref{thm:common-learning}}

To prove Theorem~\ref{thm:common-learning} we will need a lemma and a number of claims. We start with the following general lemma.
\begin{lemma}
    \label{lemma:self-sustaining_common_beliefs}

    Consider a family $(Y_i^r)_{i\in N}$ which is self-sustaining, i.e. for every $i$,
    $$Y_i^r \subseteq \birq \left (H\cap \left(\bigcap_{j\in N_i} Y_j^r\right)\right).$$
    Then, $Y_i^r \subseteq \cir$ for every $i$.
\end{lemma}
\begin{proof}
    Similar to the proof of Proposition~$\ref{prop:maximal_eqm}$,  we will prove by induction that for every agent $i$,  $Y_i^r \subseteq C_{i, r}^{q,n}$ for every $n$ and thus $Y_i^r \subseteq \bigcap_{n \geq 0} C^{q,n}_{i,r}  = \cir$.

For $n=0$, $Y_i^r \subseteq \Omega = C^{q,0}_{i,r}$. Now suppose that for every agent $j$, $Y_j^r \subseteq C_{j, r}^{k}$ for $k \leq (n -1)$. Then, $H\cap \left(\bigcap_{j\in N_i}Y_{j}^r\right) \subseteq H\cap \left(\bigcap_{j\in N_i}C_{j,r}^{n-1}\right)$. Since the belief operator is monotone and  $(Y_i^r)_{i\in N}$ are self-sustaining, 
$$Y_i^r \subseteq \birq\left (H\cap \left(\bigcap_{j\in N_i}Y_{j}^r\right)\right) \subseteq \birq\left(H\cap \left(\bigcap_{j\in N_i}C_{j,r}^{q,n-1}\right)\right) = C^{q,n}_{i,r}.$$ 
It follows that $Y_i^r \subseteq C^{q,n}_{i,r}$ for every $n$ and hence $Y_i^r\subseteq\cir$.
\end{proof}

Suppose $G$ has the similar neighborhood sizes property with constant $K$, the large information boundaries property, and subexponential-$\cQ$. Choose $\eps < (p-1/2)/(K+1)$ and $Y^r_i$ as in  \eqref{eq:self-sustaining-events}, i.e. 
$$
    Y^r_i = \left\{ S(Q) \geq (p-\eps)|Q|-\eps(\sup_j|B_r(j)|)
 \quad \forall Q\in \mathcal{Q}_i^r
\right\}.
$$ 
Since the event $Y^r_i$ is measurable with respect to $\sigma(I_i^r)$, the information available to agent $i$ with radius of observation $r$, we can identify it with a subset of $O_i^r = \{0,1\}^{B_r(i)}$, which we also denote by $Y^r_i$.

Denote $M_r = \sup_j |B_r(j)|$ and $\eta = p-\eps$, so that
    $$
        Y^r_i = \left\{ S(Q) \geq \eta|Q|-\eps M_r
 \quad \forall Q\in \mathcal{Q}_i^r
\right\}.
    $$

We first note that the set $Y^r_i \subseteq O_i^r$ is non-empty. To see this, consider the configuration $x \in O_i^r$ in which all signals are high. Then, for every $Q\in \mathcal{Q}_i^r$, we have $$S(Q) = |Q|\geq\eta|Q| - \varepsilon M_r,$$
    implying that $x\in Y^r_i$.

The next three claims establish  properties of $Y^r_i$ which are needed to show that it is self-sustaining. The first shows that on $Y_i^r$, agents learn the state, and furthermore learning is uniform over the agents and the elements of $Y_i^r$.
\begin{claim}
\label{clm:Y-1}
$$
  \lim_{r \to \infty}\inf_{i\in N}\inf_{x \in Y^r_i}\Pr(\theta = H \mid I_i^r=x)= 1.
$$
    
\end{claim}
\begin{proof}
    Consider $Q = B_r(i) \in \mathcal{Q}_i^r$. Recall from the definition of similar neighborhood sizes, there exists a constant $K$ such that for all $i, j$ and $r$ we have $|B_r(i)|/|B_r(j)|\leq K$. In particular,  $K\geq M_r/|B_r(i)|$. Then, for every $x\in Y^r_i$, we have 
    $$S(B_r(i))\geq \eta|B_r(i)| - \varepsilon M_r \geq \eta|B_r(i)| - \varepsilon K |B_r(i)|=(p-(K+1)\varepsilon)|B_r(i)|.$$
    
    Thus, for every $x \in Y_i^r$, 
    $$\frac{\Pr(\theta = H|I_i^r = x)}{\Pr(\theta = L|I_i^r = x)} = \left(\frac{p}{1-p}\right)^{2S(B_r(i)) -|B_r(i)|} \geq
\left(\frac{p}{1-p}\right)^{(2p -2(K+1)\varepsilon - 1) |B_r(i)|}.$$ 

Note that $2p -2(K+1)\varepsilon - 1 >0$ because $\varepsilon < (p-1/2)/(K+1)$, $|B_r(i)|$ is increasing in $r$ (uniformly over $i$, by  similar neighborhood sizes) and $p>1/2$. Thus, the odds ratio goes to infinity as $r$ increases. Hence, the posterior probability of the high state goes to $1$, uniformly over $i$ and $x$:
$$\lim_{r \to \infty}\inf_{i\in N}\inf_{x\in Y_i^r}\Pr(\theta = H|I_i^r = x) =1.$$

\end{proof}

The next claim shows that on $Y_i^r$, agents place high belief in $Y_j^r$ (for a neighbor $j$), and that this is again uniform over the agents and over $Y_i^r$.
\begin{claim}
\label{clm:Y-2}
$$
    \lim_{r \to \infty}\inf_{i \in N}\inf_{x \in Y_i^r}\Pr(I_j^r \in Y^r_j, \forall j \in N_i \mid I_i^r=x, \theta=H) = 1.
$$
\end{claim}
\begin{proof}
  Since the number of $i$'s neighbors is finite (and uniformly bounded over $i$, by similar neighborhood sizes for $r=1$), it suffices to fix one neighbor $j\in N_i$ and show that $$
\lim_{r \to \infty}\sup_{i \in N}\sup_{x\in Y^r_i}\mathbb{P}(I_j^r\notin Y^r_j| I_i^r=x, \theta = H)= 0.
$$  
    Consider $j\in N_i$. For each $Q\in \mathcal{Q}_j^r$, let 
    $$T = Q\cap B_r(i), \quad F = Q \backslash B_r(i).$$
    Thus, $T$ is the part of $Q$ observed by both $i$ and $j$, while $F$ is the $(j, i)$-information boundary of $Q$, i.e. the part observed by $j$ but not by $i$. Moreover, $Q = T\cup F$ and $|Q| = |T| + |F|$.


To show the desired convergence, we estimate the conditional probability that $I_j^r\not \in Y^r_j$. For this to happen, some $Q\in \mathcal{Q}_j^r$ must violate the threshold property. Thus, $$ S(Q)  < \eta|Q| - \varepsilon M_r.$$ We claim that this forces $F$ to have too few high signals, namely $S(F) < \eta|F|$. Indeed, since $x \in Y^r_i$, the threshold inequality applied to $T$ gives $$S(T) \geq \eta|T| - \varepsilon M_r.$$

Thus, if $S(F) \geq \eta|F|$, we would have 
$$
S(Q) = S(T) + S(F) \geq \eta|T| - \varepsilon M_r + \eta|F| = \eta |Q| - \varepsilon M_r,
$$
contradicting the fact that $Q$ violates the threshold inequality. Hence the event $I_j^r\notin Y^r_j$ requires some $Q\in \mathcal{Q}_j^r$ with $S(F) < \eta|F|$.

Now we show that $F$ is sufficiently large. If $Q$ violates $j$'s threshold inequality, then $0\leq S(Q) <\eta|Q| - \varepsilon M_r$ and hence 
$$
|Q|>\frac{\varepsilon}{\eta}M_r\geq \varepsilon M_r\geq \varepsilon |B_r(j)|.
$$

Now, $F = Q\setminus B_r(i)$ is either empty or  of size at least $ \delta r$, where $\delta$ is given by the large information boundaries condition. Note that $F\neq \emptyset$. Otherwise, $Q = T$, and the threshold inequality for $x\in Y_i^r$ implies that $Q$ does not violate $j$'s inequality, i.e. $S(T) = S(Q) \geq \eta |Q| - \varepsilon M_r$. Therefore $|F|\geq \delta r$.

Next, we show that within each set $F$ whose size is linearly increasing in $r$, having too few signals is exponentially unlikely. Conditional on $\theta = H$, the high signals in $F$ are independent and occur with probability $p$, so $S(F) \sim \mathrm{Bin}(|F|, p)$. Since $\eta<p$, the Chernoff bound gives
$$
\mathbb{P}\left(S(F)<\eta|F|\right)\leq \exp(-D(\eta\|p)|F|)\le
\exp(-D(\eta\|p)\delta r),
$$
where $D(\cdot||\cdot)$ is the Kullback-Leibler divergence.


Hence, by a union bound, we get 
\begin{align*}
\mathbb{P}(I_j^r\notin Y_j^r| I_i^r=x, \theta = H)
&\le
\sum_{Q\in \cQ_j^r} \mathbb{P}\left(S(Q\setminus B_r(i))<\eta|Q\setminus B_r(i)|\right)\\
&\leq
|\cQ_j^r| \cdot \exp(-D(\eta\|p)\delta r)
\to 0 \quad \text{as } r\to \infty,
\end{align*}
where the convergence holds because the social network satisfies the subexponential-$\cQ$ condition, and furthermore holds uniformly over $i$.
   
\end{proof}

The last of our three claims shows that the probability of $Y_i^r$ converges to $1$ in the high state, and that this convergence is uniform over $i$.
\begin{claim}\label{clm:Y-3}
\begin{align*}
    \lim_{r \to \infty}\inf_{i \in N}\Pr(I_i^r\in Y_i^r|\theta = H) = 1.
\end{align*}
\end{claim}

\begin{proof}
 
By a union bound and the Chernoff inequality, 
\begin{align*}
    \mathbb{P}(I_i^r\notin Y_i^r|\theta = H)
&\le
\sum_{Q\in \cQ_i^r} \mathbb{P}\left(S(Q)<\eta|Q| -  \varepsilon M_r  |\theta = H\right)\\
&=\sum_{\substack{Q\in \cQ_i^r:\\|Q| > \frac{\varepsilon}{\eta}M_r}} \Pr\left(S(Q)<\eta|Q| -  \varepsilon M_r  |\theta = H\right) \\
& \le
\sum_{\substack{Q\in \cQ_i^r:\\|Q| > \frac{\varepsilon}{\eta}M_r}} \mathbb{P}\left(S(Q)<\eta|Q| \  |\theta = H\right) \\
&\leq \sum_{\substack{Q\in \cQ_i^r:\\|Q| > \frac{\varepsilon}{\eta}M_r}}\exp(-D(\eta\|p)|Q|)\\
&\leq |\cQ_i^r| \exp\left(-D(\eta\|p)\frac{\varepsilon}{\eta}M_r\right) \\
& \leq |\cQ_i^r| \exp\left(-D(\eta\|p)\frac{\varepsilon}{\eta}(r+1)\right) .
\end{align*}
By the subexponential-$\cQ$ condition, the last expression goes to $0$ as $r$ goes to infinity, and does so uniformly over $i$.

    
\end{proof}

Fix any $q\in (0,1)$. We use the family $(Y_i^r)_{i\in N}$ constructed in \eqref{eq:self-sustaining-events} to establish network $q$-learning. First, we show that this family is self-sustaining. Using the identity $$\mathbb{P}(H, \ I_j^r\in Y^r_j \ \forall j\in N_i | I_i^r=x)
=
\mathbb{P}(H| I_i^r=x)
\mathbb{P}(I_j^r\in Y^r_j \ \forall j\in N_i | I_i^r=x, \theta = H)$$ and Claims~\ref{clm:Y-1} and~\ref{clm:Y-2}, we obtain 
$$
\lim_{r \to \infty}\inf_{i \in N}\inf_{x\in Y^r_i}\mathbb{P}(H,\ I_j^r\in Y_j^r\ \forall j\in N_i| I_i^r=x)= 1.
$$

Since $q<1$, it follows that, for all sufficiently large $r$, $Y_i^r \subseteq \birq \left (H\cap \left(\bigcap_{j\in N_i} Y_j^r\right)\right)$ for all $i$, meaning that the family $(Y_i^r)_{i\in N}$ is self-sustaining. By Lemma~\ref{lemma:self-sustaining_common_beliefs}, $Y_i^r \subseteq C_{i, r}^q$. Therefore, by Claim~\ref{clm:Y-3},
$$
    \Pr(C_{i, r}^q | \theta = H)\geq \Pr(I_i^r \in Y_i^r| \theta = H) \to 1 \quad \text{as } r\to \infty.
$$
Hence, network  $q$-learning of $H$ is obtained. Since $q\in (0,1)$ was arbitrary, network common learning of $H$ occurs. This completes the proof of Theorem~\ref{thm:common-learning}.

\subsection{Proof of Lemma \ref{lem:fixedpoint}}
\label{pf:fixedpoint}
\begin{proof}
    We first claim that the sequence of sets $(C_{j,r}^{q,n})_n$ is decreasing. This follows from the fact that the $\birq$ operator is monotone with respect to set inclusion. 

    Let $E_n = H\cap \left(\bigcap_{j\in N_i}C_{j,r}^{q,n}\right)$. Since $(C_{j,r}^{q,n})_n$ is decreasing, $(E_n)_{n\geq 0}$ is also decreasing. Hence, $E_n \downarrow \cap_{n} E_n$. Next, by continuity from above, 
    $$\Pr(\cap_{n\geq 0} E_n|I_i^r) = \lim_{n\to \infty}\Pr(E_n|I_i^r) = \inf_{n\geq 0}\Pr(E_n|I_i^r).$$
    Thus, 
    $$\omega\in \bigcap_{n\geq 0} \birq (E_n)\iff \forall n: \ \Pr(E_n|I_i^r)(\omega) \geq q \iff \inf_{n\geq 0} \Pr(E_n|I_i^r)(\omega) \geq q\iff $$
    $$\iff \Pr(\cap_{n\geq 0} E_n|I_i^r)(\omega) \geq q \iff \omega \in \birq (\cap_{n\geq 0} E_n).$$

So, $\bigcap_{n\geq 0} \birq (E_n) = \birq (\cap_{n\geq 0} E_n)$. Finally, 

    $$\cir =\Omega \cap \bigcap_{n\geq 1} C^{q,n}_{i,r}  = \bigcap_{n\geq 0} C^{q,n+1}_{i,r}  = \bigcap_{n\geq 0}\birq(E_n) = \birq (\cap_{n\geq 0} E_n) = $$
    $$= \birq \left(\bigcap_{n\geq 0} \left(H\cap \bigcap_{j\in N_i}C_{j,r}^n\right)\right) = \birq \left(H\cap \bigcap_{j\in N_i}\bigcap_{n\geq 0} C_{j,r}^n\right) = \birq \left(H\cap \bigcap_{j\in N_i} C_{j,r}\right).$$
\end{proof}

\subsection{Proof of Theorem \ref{thm:path-impossibility}}
\begin{proof}

Consider two agents $i, j$, and signal realizations $x_i \in O_{i}^r$ and $x_j \in O_{j}^{r}$ for each agent. We say that $x_i$ and $x_j$ are \emph{compatible} if they agree on the overlapping region $B_r(i) \cap B_r(j)$. 

By the definition of $C_{i,r}^{q,n}$, the event that $i$ has $n$\textsuperscript{th} order network $q$-belief, $C_{i,r}^{q,n}$ is $\sigma(I_i^r)$-measurable. That is, this event depends only on the information available to agent $i$. As such, we can identify this event with a subset of $O_i^r$, the possible observations of agent $i$. Accordingly, for $x \in O_r^i$ we write $x \in C_{i,r}^{q,n}$ whenever agent $i$ has $n$\textsuperscript{th} order network $q$-belief upon observing $x$.


By the theorem hypothesis, there is  an $(r,m)$-bottleneck path $i_1,i_2, \ldots,i_T$. Consider two adjacent agents along the bottleneck path, $i_t$ and $i_{t+1}$. Let $x_t \in O_{i_t}^r$ and $x_{t+1} \in O_{i_{t+1}}^r$ be compatible. Suppose that $x_{t+1} \not \in C_{i_{t+1},r}^{q}$, i.e., when agent $i_{t+1}$ observes $x_{t+1}$, $i_{t+1}$ does not have network $q$-belief. We claim that this implies that  $x_{t} \not \in C_{i_{t}}^{q}$, i.e., agent $i_t$ does not have  network $q$-belief upon observing the compatible $x_t$.
Indeed, since $|B_r(i_{t+1}) \setminus B_r(i_t)| \leq m$,
and the signals observed by $i_{t+1}$ but not by $i_t$ are independent conditional on the state, it follows that
\begin{align*}
    \Pr(I_{i_{t+1}}^r = x_{t+1}|I_{i_t}^r = x_t, H) = \Pr(\{s_l = x_{i_{t+1},l} \; \forall l \in \left(B_r(i_{t+1}) \setminus B_r(i_t)\right)\}|H)
    \geq (1-p)^m.
\end{align*}
By our assumption that $x_{t+1} \not \in C_{i_{t+1}}^{q}$,
$$\Pr \left(C^{q}_{i_{t+1},r} \middle\vert I^r_{i_t} = x_t, H \right) \leq 1- \Pr(I_{i_{t+1}}^r = x_{t+1}|I_{i_t}^r = x_t, H) \leq 1 - (1-p)^m,$$
and hence
$$\Pr\left(H \cap \bigcap_{j \in N_{i_t}} C^{q}_{j,r} \middle\vert I^r_{i_t} = x_t \right) \leq \Pr\left(C^{q}_{i_{t+1},r}\,\middle\vert\, I^r_{i_t}=x_t,H\right)\, \Pr(H | I^r_{i_t}=x_t) \leq$$ $$\leq 1-(1-p)^m < q.$$
Thus 
\begin{align*}
    x_t \not \in \mathcal{B}^q_{i_t,r}\left(H \cap \left(\bigcap_{j \in N_{i_t}} C^{q}_{j,r}\right)\right) =C^{q}_{i_t,r},
\end{align*}
by the fixed-point property of $C^q_{i,r}$ (Lemma~\ref{lem:fixedpoint}). We have thus shown that \begin{align}
    \label{eq:cascade}
    x_{t+1} \not \in C^q_{i_{t+1},r}\implies x_t \not \in C^q_{i_{t},r}.
\end{align}

Fix an arbitrary observation $x_1 \in O_{i_1}^r$ for agent $i_1$. Since $B_r(i_1)\cap B_r(i_T) = \emptyset$, there exists a signal configuration on $\bigcup_t B_r(i_t)$ that agrees with $x_1$ on $B_r(i_1)$ and with $0$ (the all zeros observation) on $B_r(i_T)$. Hence, there are $x_2,\ldots,x_T$ such that $x_T=0$ and $x_1,x_2,\ldots,x_T$ are compatible. Note that $0 \not \in C^{q}_{i_T,r}$, since when an agent observes all zero signals, their posterior belief in $H$ is below one half, and hence below $q$.

Thus, using \eqref{eq:cascade}, it follows by induction that $x_1 \notin C^{q}_{i_1,r}$:
$$
x_{T} \notin C^{q}_{i_T,r} \implies x_{T-1} \notin C^{q}_{i_{T-1},r} \implies \cdots \implies x_1 \notin C^{q}_{i_1,r}
$$
Since the choice of $x_1$ is arbitrary, we have shown that $C_{i_1,r}^{q} = \emptyset$: $i_1$ never has network  $q$-belief. It then follows from the fixed point property of these sets (Lemma~\ref{lem:fixedpoint}), that the same holds for $i_1$'s neighbors, and by induction for all the agents.

Finally, by Proposition~\ref{prop:maximal_eqm}, it follows that all agents play the safe action in every equilibrium.
\end{proof}

\bibliography{refs}

\end{document}